\documentclass[conference,letterpaper]{IEEEtran}

\usepackage[utf8]{inputenc} 
\usepackage[T1]{fontenc}
\usepackage{url}
\usepackage{ifthen}
\usepackage{cite}
\usepackage[cmex10]{amsmath} 

\usepackage{xcolor}
\usepackage{mathtools}
\mathtoolsset{showonlyrefs}

\usepackage{epsfig,amsfonts,subfigure,color,amsbsy}
\usepackage[nolist]{acronym}
\usepackage{graphicx,cite,amssymb,amsthm}
\usepackage{amsmath}
\usepackage{mathrsfs}
\usepackage{tikz}
\usepackage{bm}
\usepackage{multirow}
\usepackage{bigstrut}
\usepackage{psfrag}
\usepackage{stackengine}
\usetikzlibrary{shapes,arrows}
\usepackage{xcolor}
\usepackage{mathtools}
\usepackage{tabularx,array}
\usepackage{mathbbol}
\usepackage{lipsum}

\usepackage{tikz}
\usetikzlibrary{shapes,arrows}
\usetikzlibrary{calc}
\usetikzlibrary{patterns} 
\usepackage{mathrsfs}
\usetikzlibrary{fadings}
\usetikzlibrary{shadows.blur}
\usetikzlibrary{spy}
\usetikzlibrary{positioning,chains,fit}

\theoremstyle{theorem}
\newtheorem{theorem}{Theorem}
\newtheorem{lemma}{Lemma}
\newtheorem{proposition}{Proposition}
\newtheorem{corollary}{Corollary}
\newtheorem{problem}{Problem}
\theoremstyle{definition}
\newtheorem{definition}{Definition}

\newtheorem{example}[theorem]{Example}

\newcommand{\st}{\, | \,} 

\newcommand{\vect}[1]{\boldsymbol{#1}}

\newcommand{\bfx}{\vect{x}}
\newcommand{\bfy}{\vect{y}}

\newcommand{\bfI}{\vect{I}}

\newcommand{\card}[1]{\left\vert{#1}\right\vert} 
\newcommand{\floor}[1]{\lfloor{#1}\rfloor} 
\newcommand{\ceil}[1]{\left\lceil{#1}\right\rceil} 
\newcommand{\set}[1]{\left\lbrace{#1}\right\rbrace} 

\newcommand{\NN}{\mathbb{N}}

\newcommand{\alphabet}{\mathcal{A}}
\newcommand{\colorsetN}[2]{\mathcal{S}^{(n)}_{#1}(#2)}

\newcommand{\setcoloring}{\mathcal{I}}
\newcommand{\setoutput}{\mathcal{S}}

\newcommand{\entropy}[1]{\mathrm{H}\left(#1\right)}

\newcommand{\tmin}{t_{\min}}
\newcommand{\Tmin}{T_{\min}}

\newcommand{\rateN}{\mathcal{R}_n}
\newcommand{\capacity}{\mathcal{C}}

\definecolor{myteal}{RGB}{0,128,106}
\definecolor{myblue}{RGB}{69, 220, 180}
\definecolor{myred}{RGB}{222, 0, 49}
\definecolor{myyellow}{RGB}{255,170,0}

\colorlet{greyblue}{cyan!70!magenta}

\newcommand{\JB}[1]{\textcolor{greyblue!80!green}{[JB: #1]}}

\newcommand{\ey}[1]{\textcolor{blue}{[EY: #1]}}

\begin{document}
\title{Sequence Reconstruction over Coloring Channels for Protein Identification}


\author{%
  \IEEEauthorblockN{Jessica Bariffi}
  \IEEEauthorblockA{Department of Computer Engineering\\
                    Technical University of Munich\\
                    Munich, Germany\\
                    Email: jessica.bariffi@tum.de}
  \and
  \IEEEauthorblockN{Antonia Wachter-Zeh}
  \IEEEauthorblockA{Department of Computer Engineering\\
                    Technical University of Munich\\
                    Munich, Germany\\
                    Email: antonia.wachter-zeh@tum.de}

    \and
  \IEEEauthorblockN{Eitan Yaakobi}
  \IEEEauthorblockA{Departement of Computer Science\\ 
                    Technion -- Israel Institute of Technology\\
                    Haifa, Israel\\
                    Email: yaakobi@cs.technion.ac.il}
}

\maketitle
\begin{abstract}
   This paper studies the sequence reconstruction problem for a  channel inspired by protein identification. We introduce a \emph{coloring channel}, where a sequence is transmitted through a channel that deletes all symbols not belonging to a fixed subset (the coloring) of the alphabet. By extending this to a \emph{coloring profile}, a tuple of distinct colorings, we analyze the channel’s information rate and capacity. We prove that optimal (i.e., achieving maximum information rate) coloring profiles correspond to $2$-covering designs and identify the minimal covering number required for maximum information rate, as well as the minimum number for which any coloring profile is optimal.
\end{abstract}

\section{Introduction}\label{sec:intro}
Motivated by various applications in bioinformatics, molecular biology, and biochemistry, the large family of sequence reconstruction problems (such as, for instance, \textit{trace reconstruction} \cite{batu2004reconstructing} or \textit{Levenshtein's reconstruction problem}\cite{levenshtein2001efficient,levenshtein2001efficientSUBSUPER}) investigate in the challenge of recovering a sequence from a sufficient number of noisy copies. One of these families of sequence reconstruction considers a channel where the output is not given by a noisy version of the transmitted word, but is given by a list of specific subsequences. The goal there is to uniquely recover the original sequence given the subsequences at the channel output.
Sequence reconstruction problems have been widely studied for various channel models such as substitutions, deletions, insertions, and repetitions \cite{levenshtein2001efficient,levenshtein2001efficientSUBSUPER,abu2021levenshtein,magner2016fundamental,shafir2023sequence,yehezkeally2021uncertainty,sabary2024reconstruction,sala2017exact}. As a result, they have gained significant attention, particularly in the field of DNA sequencing.

While modern DNA sequencing technologies have achieved single-molecule resolution (e.g., in nanopore sequencing), protein sequencing still relies on bulk averaging across multiple cells (see \cite{bekker2017optimized} for more details). Consequently, protein identification remains a challenging task. In~\cite{ohayon2019simulation}, a novel method for single-protein identification using tricolor fluorescence was introduced, with its feasibility analyzed at a computational level. This approach involves labeling three amino acids in the protein chain with corresponding fluorescence markers and then translocating the protein through a nanopore while it is excited by a laser. The resulting output is a tricolored trace chart, which identifies the protein with approximately $96$\% accuracy compared to a database of human proteins.

In this paper, we study this model from an information-theoretical point of view. To model the process of associating a fluorescence marker with a specific protein, we consider an alphabet of a fixed size and subset of that alphabet. We refer to such a subset as a \textit{coloring}. We define an error-free channel based on this coloring, which takes as input a length-$n$ sequence over the alphabet. The channel's output is a subsequence that consists only of entries belonging to the coloring. We refer to this as a coloring channel, which can be interpreted as a deletion-like channel, where all symbols not belonging to the coloring are deleted.
For instance, consider the ternary sequence $\bfx = (012221001)$ and the coloring $I = \set{0,1} \subset \set{0, 1, 2}$. Transmitting the sequence $\bfx$ over the coloring channel with respect to $I$ yields the output sequence $\bfy_{I}(\bfx) = (011001)$.
Extending this idea, we consider a $t$-tuple of distinct colorings of the same size, which we call a \textit{coloring profile}. We define a channel with respect to a coloring profile similarly: the output is a $t$-tuple of subsequences obtained from the input by deleting all entries that do not belong to the colorings in the profile. Figure \ref{fig:diagram_coloringsChannel} illustrates this model given a $t$ coloring profile.
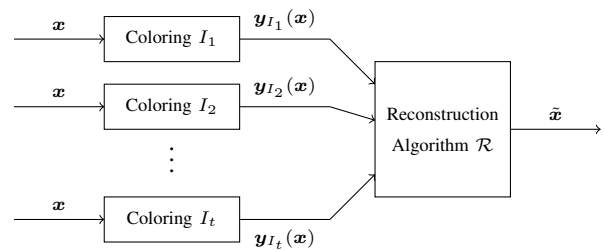
\begin{figure}[h]
    \centering
    \vspace{-2mm}
    \begin{tikzpicture}[scale = 1.2]
        \foreach \i in {2, 3.25, 4}{
            \draw[->] (0,\i) -- (1, \i);
            \node[above] at (0.5, \i) {\scriptsize $\bfx$};
            \draw (1, \i - 0.25) rectangle (2.5, \i + 0.25);
            \draw (2.5, \i) -- (3.5, \i);  
            }
        \node at (1.75, 2) (Chan) {\scriptsize Coloring $I_t$};
        \node at (1.75, 3.25) (Chan) {\scriptsize Coloring $I_2$};
        \node at (1.75, 4) (Chan) {\scriptsize Coloring $I_1$};
        \node at (1.75, 2.75) {$\vdots$};s
        \draw[->] (3.5, 4) -- (4, 3.5);
        \node[above] at (3, 4) {\scriptsize $\bfy_{I_1}(\bfx)$};
        \draw[->] (3.5, 3.25) -- (4, 3.1);
        \node[above] at (3, 3.25) {\scriptsize $\bfy_{I_2}(\bfx)$};
        \draw[->] (3.5, 2) -- (4, 2.5);
        \draw (4, 2.25) rectangle (5.5, 3.75);
        \node[below] at (3, 2) {\scriptsize $\bfy_{I_t}(\bfx)$};
        \node[above] at (4.75, 3) {\scriptsize Reconstruction};
        \node[below] at (4.75, 3) {\scriptsize Algorithm $\mathcal{R}$};
        \draw[->] (5.5,3) -- (6.5, 3);
        \node[above] at (6,3) {\scriptsize $\Tilde{\bfx}$};
    \end{tikzpicture}
    \vspace{-2mm}
    \caption{Illustration of reconstructing a sequence $\bfx$ over $t$ distinct $c$-coloring channels.}
    \label{fig:diagram_coloringsChannel}
\end{figure}
\vspace{-2mm}
The goal under this paradigm is to reconstruct the input sequence given the outputs of the channels. In order to do so, we address three main problems. 
First, for given parameters of the channel model, we are interested in the information rate and capacity. 
Second, we characterize the coloring profiles for which the maximum information rate is achieved. We prove that these coloring profiles are the $2$-covering designs (see \cite{assmus1992designs,gordon1995new,mills1992coverings} for an overview of covering designs). The minimal number for which a $2$-covering design (and hence a coloring profile achieving maximum information rate) exists is referred to as the \textit{minimal covering number}, which has been extensively studied in the past \cite{assmus1992designs,bose1939construction,gordon1995new,mills1992coverings,moore1896tactical,schonheim1964coverings,etzion1995bounds,nurmela1993upper}. In conclusion, we state the minimum size at which any coloring profile can be considered a $2$-covering design.
\section{Problem Statements}\label{sec:problems}
We start by fixing the notation needed throughout the paper. For some $p \in (0, 1)$ we define the binary entropy function as
\begin{align}\label{eq:binary_entropy}
    \entropy{p} := -p\log_2(p) - (1-p)\log_2(1-p).
\end{align}
It is well-known that the binomial coefficient $\binom{n}{k}$, for two nonnegative integers $n, k$ with $n\geq k$, can be bounded as
\begin{align}\label{eq:bounds_binomial-coefficient}
    \frac{1}{n+1}2^{n\entropy{k/n}} \leq \binom{n}{k} \leq 2^{n\entropy{k/n}} .
\end{align}

In the following, for a positive integer $q$, we consider a $q$-ary alphabet which we denote by $\alphabet_q := \set{0, \ldots, q-1}$. We use italic-bold letters for vectors to distinguish them from variables. For a positive integer $c$, we refer to a subset of $\alphabet_q$ of cardinality $c$ as a \emph{$c$-coloring}. Furthermore, we denote by $\setcoloring(q, c)$ the set of all $c$-colorings over $\alphabet_q$.
Notice that $\card{\setcoloring(q, c)} = \binom{q}{c}$.

Given a sequence $\bfx \in \alphabet^n$ and a $c$-coloring $I$, we denote the subsequence of $\bfx$ obtained by deleting all entries that are not contained in $I$ by $\bfy_{I}(\bfx)$ and we call $\bfy_{I}(\bfx)$ the \textit{$I$-colored subsequence of $\bfx$}. 
\begin{example}
    Consider the ternary alphabet $\alphabet_3 = \set{0, 1, 2}$, a sequence $\bfx = (1, 2, 0, 0, 2, 0, 1, 1, 0) \in \alphabet_3^9$ and the $2$-colorings $I_1 = \set{0, 1}$ and $I_2 = \set{1, 2}$. Then the following $I_1$-, respectively, $I_2$-colored subsequences of $\bfx$ are given by
    \begin{align}
        \bfy_{I_1}(\bfx) = (1, 0, 0, 0, 1, 1, 0), \; \text{and} \; \bfy_{I_2}(\bfx) = (1, 2, 2, 1, 1).
    \end{align}
\end{example}
Let $I_1, \ldots, I_t \in \setcoloring(q,c)$ be distinct colorings. We call the $t$-tuple $\bfI := (I_1, \ldots, I_t)$ a \textit{$c$-coloring profile of size $t$} of $\alphabet_q$.
For a given sequence $\bfx \in \alphabet_q^n$ and a $c$-coloring profile $\bfI$ of size $t$, we define the \textit{$\bfI$-colored subsequence of $\bfx$} as 
$$\bfy_{\bfI}(\bfx) := (\bfy_{I_1}(\bfx), \ldots , \bfy_{I_t}(\bfx)).$$
Furthermore, we denote the set of all $\bfI$-colored subsequences as
$$\setoutput_q^{(n)}(\bfI) := \set{ \bfy_{\bfI}(\bfx) \st \bfx \in \alphabet_q^n }.$$
Note that $\setoutput_q^{(n)}(\bfI)$ may contain the empty sequence if none of the entries of $\bfx$ lie in $\bfI$.\\

We now introduce a channel model over a $q$-ary alphabet $\alphabet_q$ based on a fixed $c$-coloring profile $\bfI$ of size $t$ where, for any input $\bfx \in \alphabet_q^n$, the output is defined by the subsequence $\bfy_{\bfI}(\bfx)$ (see Figure \ref{fig:diagram_coloringsChannel}). Hence, this channel model can be seen as a deletion-like channel, where all symbols not lying in the coloring profile $\bfI$ are deleted. Besides the deletions due to the colorings we do not consider any other error in the transmission of the input sequence $\bfx \in \alphabet_q^n$. The goal then is to uniquely reconstruct $\bfx$ from the $\bfI$-colored sequence $\bfy_{\bfI}(\bfx)$. 
If a sequence $\bfx \in \alphabet_q^n$ can be uniquely reconstructed from $\bfy_{\bfI}(\bfx)$, we say that the sequence $\bfx$ is \textit{$\bfI$-reconstructible}. For instance, if $c = q$ and $\bfI = \alphabet_q$, we obtain $\bfy_{\bfI}(\bfx) = \bfx$ for every $\bfx\in\alphabet_q^n$. Thus, every sequence is $\bfI$-reconstructible. On the other hand, if $c = 1$ and $\bfI$ is a $1$-coloring profile of size $t$, we can only recover sequences $\bfx \in I_j^n$ for $j = 1, \ldots, t$. Generally $\bfy_{\bfI}(\bfx)$ only tells us the number of elements in $\bfx$ contained in each of the $I_j$´s, but no information about the order of the entries can be revealed.

Furthermore, we define the joint channel information rate and capacity, respectively, by
\begin{align}\label{eq:infoRate}
    \rateN(q, c, \bfI) &:= \frac{1}{n}\log_q\left(\big|\setoutput_q^{(n)}(\bfI)\big|\right), \\
    \capacity(q, c, \bfI) &:= \limsup_{n \rightarrow \infty} \rateN(q, c, \bfI) .
\end{align}
If $q$ and $c$ are clear from the context, we will write $\rateN(\bfI)$ and $\capacity(\bfI)$. Additionally, since $\big| \setoutput_q^{(n)} (\bfI) \big| \leq \card{\alphabet_q^n }= q^n$ for any $\bfI$, we have that $\rateN(\bfI), \capacity(\bfI) \leq 1$. \\

The rate and the capacity of $t$ joint $c$-coloring channels allow us to understand the reconstruction performance given any input sequence $\bfx \in \alphabet_q^n$. Hence, the first problem we study is understanding the information rate for a given coloring profile.
\begin{problem}\label{problem:info_rate}
    Given a $c$-coloring profile $\bfI = (I_1, \ldots, I_t)$ of length $t$ over a $q$-ary alphabet $\alphabet_q$, and a positive integer $n$, find the information rate $\rateN (\bfI)$ and capacity $\capacity(\bfI)$.
\end{problem}
We study Problem \ref{problem:info_rate} in Section \ref{sec:inforate}, especially focussing on colorings of size $q-1$ and colorings that are mutually disjoint. Since ideally, given a $c$-coloring profile $\bfI$, we would like that almost any $\bfx \in \alphabet_q^n$ is $\bfI$-reconstructible, the second problem we study is the following.
\begin{problem}\label{problem:coloring_achieving_maxCapacity}
    Given an alphabet size $q$ and two positive integers $c$ and $t$. Characterize the $c$-coloring profiles $\bfI = (I_1, \ldots, I_t)$ that achieve maximum information rate $\rateN(\bfI) = 1$, i.e., $\card{\setoutput_q(\bfI)} = \card{\alphabet_q^n} = q^n$, for all $n \in \mathbb{N}$.
\end{problem}
Note that maximum capacity is also reached if a coloring profile $\bfI$ achieves maximum information rate. The contrary, however, is not necessarily true.
Once an answer is found to Problem \ref{problem:coloring_achieving_maxCapacity}, it is interesting to understand the minimum, maximum, or average number of colorings needed to fulfill the characterization properties. That leads to a third and last problem.
\begin{problem}\label{problem:tmin}
    Given an alphabet size $q$ and a positive integer $c \geq 2$. Assume that we draw $t$ distinct colorings $I_1, \ldots, I_t \in \setcoloring(q, c)$ and let $\bfI$ denote the corresponding $c$-coloring profile.
    \begin{enumerate}
        \item Find the smallest number $\tmin(q, c)$ for which it exists a $c$-coloring profile $\bfI$ of size $\tmin(q, c)$ that achieves the maximum information rate?
        \item Find the smallest number $\Tmin(q, c)$ such that any $c$-coloring profile $\bfI$ of size $\Tmin(q,c)$ achieves maximum information rate. 
    \end{enumerate}
\end{problem}
Note that, trivially, if $c = q$ the channel is deletion-free (i.e., $\bfy(\bfx) = \bfx$) and hence $\tmin(q,q) = \Tmin(q,q) = 1$. Additionally, if $c=1$ and $q\geq2$, we can never reconstruct all the sequences over the alphabet $\alphabet_q$. We discuss
Problems \ref{problem:coloring_achieving_maxCapacity} and \ref{problem:tmin} in more detail in Section \ref{sec:max_capacity}.
\section{Information Rates and Capacities}\label{sec:inforate}
In this section, we study the information rate and capacity of coloring channels for given parameters $q, c$, and $t$. That is, we assume that over a $q$-ary alphabet, the size $c$ of the colorings and the number $t$ of $c$-colorings are given, and we discuss Problem \ref{problem:info_rate} for these $c$-coloring channels over $\alphabet_q$ (see \eqref{eq:infoRate} for the definition). That is we compute the information rate and channel capacity given the parameters $q, c$ and $t$, which boils down to computing the size of the set $\setoutput_q^{(n)}(\bfI)$. Recall that for $c = 1$ and any $1$-coloring profile $\bfI = (I_1, \ldots, I_t)$, we obtain $\big| \setoutput_q^{(n)} (\bfI)\big| = \card{\cup_{i = 1}^t\set{i = 1, \ldots , n \st x_i \in I_j, \; \bfx\in\alphabet_q^n}}$. Additionally, for $c = q$ each $c$-coloring profile $\bfI$ of size $t$ is given by $\bfI = (\alphabet_q, \ldots, \alphabet_q)$ and thus $\big| \setoutput_q^{(n)} (\bfI)\big| = \big| \alphabet_q \big| =q^n$.

In the following we will focus only on $c$-colorings with $1<c<q$.
Let us start with the case where we consider one $c$-colorings channel. Equivalently, we can think of a deletion channel where all symbols of exactly $d = q - c$ symbols are removed.
\begin{lemma}\label{lemma:one_(q-1)-coloring}
    Given $I \in \setcoloring(q, c)$, we observe that
    \begin{align}
        \big|\colorsetN{q}{I} \big| = \sum_{i = 0}^n c^{i} = \frac{c^{n+1}-1}{c-1}.
    \end{align}
    In particular, it holds that
    \begin{align}
        \rateN(q, c, I) &= \tfrac{1}{n}\log_q(c^{n+1}-1) + \tfrac{1}{n}\log_q(c - 1),\quad \text{and} \\
        \capacity(q, c, I) &= \log_q(c).
    \end{align}
\end{lemma}
\begin{proof}
    Observe that the number of output sequences depends only on the number $i$ of positions that are not equal to the deleted values. For each of these positions, there are $c$ options, and hence the expression for $\big|\colorsetN{q}{I} \big|$ follows. We obtain the information rate and capacity by calculating the geometric sum and suitably bounding it from below and above.
\end{proof}

We now focus on $t \geq 2$ and $c=q-1$.
\begin{proposition}\label{prop:two_(q-1)colorings}
    Let $n$ and $q$ be two positive integers and consider the $q$-ary alphabet $\alphabet_q$. For any two distinct colorings $I_1, I_2 \in \setcoloring(q, q-1)$, it holds that
    \begin{align}
        \big| \colorsetN{q}{I_1, I_2} \big| = \sum_{i = 0}^n \sum_{j = 0}^{n-i} \binom{n-i}{j}\binom{n-j}{i}(q-2)^{n-i-j} .
    \end{align}
\end{proposition}
\begin{proof}
    Let $a_1, a_2 \in \alphabet_q$ be distinct and assume $I_i = \alphabet_q \setminus \set{a_i}$ for $i = 1, 2$. Given a sequence $\bfx \in \alphabet_q^n$, we consider the output $\bfy_{\bfI}(\bfx) = (\bfy_{I_1}(\bfx), \bfy_{I_2}(\bfx))$, where $\bfy_{I_i}(\bfx)$ is obtained by the entries of $\bfx$ lying in $I_i$. In other words, $\bfy_{I_i}$ is obtained by deleting all entries of $\bfx$ equal to $a_i$. Since $\bfy(\bfx)$ is a concatenation, we essentially consider $d = 2$ deleted symbols ($a_1$ for the first part of $\bfy(\bfx)$ and $a_2$ for the second).
    
    Assume that we consider an input sequence $\bfx \in \alphabet_q^n$ consisting of $i$ entries equal to $a_1$, $j$ entries equal to $a_2$, and $n-i-j$ entries lying in $\alphabet_q\setminus\set{a_1, a_2}$. 
    Since, for $\bfy_{I_1}(\bfx)$, all entries of $\bfx$ that are equal to $a_1$ are deleted, there are exactly
    \begin{align}
        \binom{n-i}{j}
    \end{align}
    possibilities of combining $j$ elements equal to $a_j$ and $n-i-j$ elements in $\alphabet_q\setminus\set{a_1, a_2}$.
    Similarly, for $\bfy(\bfx)_{I_2}$, there are 
    \begin{align}
        \binom{n-j}{i}
    \end{align}
    such possibilities.
    Additionally, for each of the $n-i-j$ elements in $\alphabet_q\setminus\set{a_1, a_2}$ there are $q-d$ options. Hence, summing over all possibilities for $i$ and $j$ we obtain the desired result.
\end{proof}

We can now study the capacity of two distinct $(q-1)$-coloring channels.

\begin{theorem}\label{thm:limit_2col_1del_qary}
    Let $q \geq 3$ and denote $s := \frac{1}{2} - \frac{\sqrt{(q+2)(q-2)}}{2(q+2)}$. For any two distinct colorings $I_1, I_2 \in \setcoloring(q, q-1)$ it holds
    \begin{align}
        \capacity(q, &q-1, (I_1, I_2))\\
        &= \log_q(4) (1-s)\entropy{\frac{s}{1-s}} + (1-2s)\log_q(q-2).
    \end{align}
\end{theorem}
In particular, we have that $\lim_{q \rightarrow \infty} \capacity(q, q-1, (I_1, I_2)) = 1$.
\begin{proof}
    Recall from Proposition \ref{prop:two_(q-1)colorings} that we have
    \begin{align}
        \big| \colorsetN{q}{I_1, I_2} \big| = \sum_{i = 0}^{n} \sum_{j = 0}^{n-i} \binom{n-i}{j}\binom{n-j}{i} (q-2)^{n-i-j}.
    \end{align}
    We now consider $i = s_in$ and $j = t_jn$ where $0 < s_i,t_j < 1$. Using the upper bound in \eqref{eq:bounds_binomial-coefficient} yields
    \begin{align}
        \big| \colorsetN{q}{I_1, I_2} \big| 
        \leq 
        \sum_{i = 0}^n \sum_{j = 0}^{n-i}\Big( & 2^{n(1-s_i)\entropy{ \frac{t_j}{1-s_i} } + n(1-t_j)\entropy{ \frac{s_i}{1-t_j} }}\\
        &\quad (q-2)^{n(1-s_i-t_j)}\Big).
    \end{align}
    Let $s, t \in (0, 1)$ denote the values maximizing, for any $i, j$, $$2^{n(1-s_i)\entropy{ \frac{t_j}{1-s_i} } + n(1-t_j)\entropy{ \frac{s_i}{1-t_j} }}(q-2)^{n(1-s_i-t_j)}.$$ This yields the following upper bound  on $\big| \colorsetN{q}{I_1, I_2} \big|$
    \begin{align}
        2^{n(1-s)\entropy{ \frac{t}{1-s} } + n(1-t)\entropy{ \frac{s}{1-t} }}(q-2)^{n-s-t} \tfrac{(n+1)(n+2)}{2} ,
    \end{align}
    and therefore
    \begin{align}\label{eq:prop:two_(q-1)colorings:limit_upper}
        \capacity&(q, q-1, (I_1, I_2))\\
        & \leq
        \log_q(2)\left( (1-s)\entropy{ \tfrac{t}{1-s} } + (1-t)\entropy{ \tfrac{s}{1-t}} \right) \\
        &\quad\quad + (1-s-t)\log_q(q-2).
    \end{align}
    On the other hand, we can lower bound the sum of binomial coefficients by using one term of the double sum only. For instance, the term where $i = \floor{sn}$ and $j = \floor{tn}$. Then we can lower bound $\big| \colorsetN{q}{I_1, I_2} \big|$ by
    \begin{align}\label{eq:prop:two_(q-1)colorings:lowerbound_size}
        \binom{n - \floor{tn}}{\floor{sn}}\binom{n-\floor{sn}}{\floor{tn}}(q-2)^{n- \floor{sn}-\floor{tn}}.
    \end{align}
    Applying the lower bound \eqref{eq:bounds_binomial-coefficient} on the binomial coefficient and the fact that $\floor{tn} \leq tn$ yields
    \begin{align}
        \binom{n - \floor{tn}}{\floor{sn}} 
        &\geq \tfrac{1}{(n+1)}2^{(n-tn)\entropy{\tfrac{\floor{sn}}{n- \floor{tn}}}},
    \end{align}
    where we used that $\floor{x} \leq x$ for any $x \in \mathbb{R}$.
    Similarly, we can show that $\entropy{\tfrac{\floor{sn}}{n- \floor{tn}}}$ is lower bounded by
    \begin{align}
        -\tfrac{s}{1-t}\log_2(\tfrac{s}{1-t}) - (\tfrac{1-s-t+2/n}{1-t})\log_2(\tfrac{1-s-t+2/n}{1-t}).
    \end{align}
    A similar lower bound is found for the coefficient $\binom{n - \floor{sn}}{\floor{tn}}$ and for $\entropy{\tfrac{\floor{tn}}{n- \floor{sn}}}$.
    Using these further lower bounds, we obtain that
    \begin{align}\label{eq:prop:two_(q-1)colorings:limit_lower}
        \capacity&(q, q-1, (I_1, I_2))\\
        &\geq
        \log_q(2)\left( (1-s)\entropy{ \tfrac{t}{1-s} } + (1-t)\entropy{ \tfrac{s}{1-t}} \right) \\
        &\quad\quad + (1-s-t)\log_q(q-2).
    \end{align}
    Combining the bounds in \eqref{eq:prop:two_(q-1)colorings:limit_upper} and \eqref{eq:prop:two_(q-1)colorings:limit_lower} yields the expression for the limit.
    By solving a maximization problem using Lagrange multipliers, we find $s = t = \tfrac{1}{2} - \frac{\sqrt{(q+2)(q-2)}}{2(q+2)}$.
\end{proof}
For simplicity, from now on we exclusively focus on $c$-coloring profiles $\bfI$ of size $t$ whose $c$-colorings are pairwise disjoint. We refer to such a coloring profile as a \textit{disjoint coloring profile}. 

As a first case, we further restrict to a disjoint coloring profile covering the whole alphabet.
\begin{theorem}\label{thm:4ary_disjoint}
    Let $q \geq 4$ and let $t, c$ be two positive integers. Let $\bfI$ be a disjoint $c$-coloring profile of size $t$. The number of $\bfI$-colored sequences is given by
    \begin{align}
        \big| \setoutput_q^{(n)}(\bfI) \big| = c^n\binom{n+t-1}{t-1},
    \end{align}
    if $ct = q$, and 
    \begin{align}
        \big| \setoutput_q^{(n)}(\bfI) \big| = \sum_{i=0}^n c^i\binom{i+t-1}{t-1},
    \end{align}
    if $ct<q$.
    In particular, in both cases, $\capacity(q, c, \bfI) = \log_q(c)$. 
\end{theorem}
\begin{proof}
    We start by considering the case where $tc = q$.
    By assumption, we have $I_1 \sqcup I_2 \sqcup \cdots \sqcup I_t = \alphabet_q$, where $\sqcup$ denotes the disjoint union. Hence, assume that the channel input sequence $\bfx \in \alphabet_q^n$ consists of $i_1$ entries that lie in $I_1$, $i_2$ entries that lie in $I_2$, and so forth. Since all colorings in the profile $\bfI$ are pairwise disjoint, there are $c^{i_j}$ distinct outputs for $\bfy_{I_j}(\bfx)$ for each coloring $I_j$ where $j = 1, \ldots, t$.
    
    Notice that $i_1 + \cdots + i_t = n$. 
    By the ordering of the colorings in the profile $\bfI$,
    the number of possibile outputs is given by 
    \begin{align}
        \setoutput_q^{(n)}(\bfI) = c^n \binom{n+t-1}{t-1}.
    \end{align}
    The limit is easily computed by taking the logarithm base $q$ and dividing by $n$.

    Let now $tc < q$ and assume that an input sequence $\bfx \in \alphabet_q^n$ consists of $i$ entries that are contained in $I_1\sqcup \cdots \sqcup I_t$  and $n-i$ entries in none of the $c$-colorings where $i = 1, \ldots , n$. Then the expression is analogous to the case above considering $tc = i$. Summing over all values of $i$ yields the desired result for the number of $\bfI$-colored sequences.
    To obtain the capacity, we observe that 
    \begin{align}
        c^n \binom{n + t-1}{t-1} \leq \big| \setoutput_q^{(n)}(\bfI)\big| \leq n c^n \binom{n + t-1}{t-1}.
    \end{align}
    Applying the logarithm base $q$, dividing by $n$ and applying the limit gives $\capacity(q, c, \bfI) = \log_q(c)$
\end{proof}
\section{Channels Achieving Maximum Rate}\label{sec:max_capacity}
In this section, we address Problems \ref{problem:coloring_achieving_maxCapacity} and \ref{problem:tmin}. As a first step, we characterize the $c$-coloring channels that achieve the maximum capacity. That is, we focus on $t$ distinct channel outputs of a $c$-coloring channel over $\alphabet_q$ and give a necessary and sufficient condition on the colorings for which any input $\bfx \in \alphabet_q^n$ is $\bfI$-reconstructible. As the cases $c = 1$ and $c = q$ are trivial, we restrict to the cases where $2 \leq c \leq q-1$.
\begin{theorem}\label{thm:covering_pairs}
    Consider a $c$-coloring profile $\bfI$ of $t$ distinct colorings over $\alphabet_q$. Any input sequence $\bfx \in \alphabet_q^n$ is $\bfI$-reconstructible if and only if every pair $(a, b)\in\alphabet_q^2$ is contained in at least one coloring of the profile $\bfI$.
\end{theorem}
\begin{proof}
    First of all, let us assume that every pair of distinct points in $\alphabet_q$ is incident to at least one $c$-coloring in $\bfI$. This implies that $\alphabet_q = \bigcup_{i = 1}^t I_i$. We show inductively that for any distinct $\bfx, \bfx' \in \alphabet_q^n$ it holds that $\bfy_{\bfI}(\bfx) \neq \bfy_{\bfI}(\bfx')$. 
    Thus, consider the base case $n = 1$. That is, $x, x' \in \alphabet_q$ with $x \neq x'$. Furthermore, there exists $i \in \set{1, \ldots , t}$ such that $x, x' \in I_i$. Since $x$ and $x'$ are distinct, we observe $y_{I_i}(x) \neq y_{I_i}(x')$ and therefore $\bfy_{\bfI}(x) \neq \bfy_{\bfI}(x')$.
    Now assume that the statement holds for any arbitrary but fixed $n\geq 1$ and let $\bfx, \bfx' \in \alphabet_q^{n+1}$ be distinct. Without loss of generality, we assume that $\bfx$ and $\bfx'$ can be written as
    \begin{align}
        \bfx = (a, \bfx^{(n)}), \quad \text{and} \quad \bfx' = (a', \bfx'^{(n)})
    \end{align}
    where $a, a' \in \alphabet_q$ and $\bfx^{(n)}, \bfx'^{(n)} \in \alphabet_q^n$. Note that there is a $c$-coloring $I_i$, for $i = 1, \ldots , t$, such that $a, a' \in I_i$. If $a \neq a'$ it holds $y_{I_i}(a) = a \neq a' = y_{I_i}(a')$ and thus, $\bfy_{\bfI}(\bfx) \neq \bfy_{\bfI}(\bfx')$. On the other hand, if $a = a'$, we must have that $\bfx^{(n)} \neq \bfx'^{(n)}$. By the induction hypothesis it follows that $\bfy_{\bfI}(\bfx^{(n)}) \neq \bfy_{\bfI}(\bfx'^{(n)})$ and therefore also $\bfy_{\bfI}(\bfx) \neq \bfy_{\bfI}(\bfx')$.

    On the other hand, assume that there are two points $a, b \in \alphabet_q$ for which there is no $c$-coloring $I \in \bfI$ containing both elements. This implies the following three scenarios
    \begin{enumerate}
        \item There exist distinct $i, j \in \set{1, \ldots, t}$ such that $a \in I_i$ and $b \in I_j$, but $a \not\in I_j$ and $b \not \in I_i$.
        \item There is no $i \in \set{1, \ldots, t}$ such that $a \in I_i$.
        \item There are no $i,j \in \set{1, \ldots, t}$ such that $a \in I_i$ and $b \in I_j$.
    \end{enumerate}
    In the first case, without loss of generality, say $a \in I_1$, $b \in I_2$, $\bfx = (a, b)$ and $\bfx' = (b, a)$. Then, however, $\bfy_{\bfI}(\bfx) = (a, b) = \bfy_{\bfI}(\bf\bfx')$. Considering the second case, the argument is similar. For instance, if $\bfx$ and $\bfx'$ both consist of $k$ entries equal to $a$ with $a \in I_i$ for some $i \in \set{1, \ldots, t}$ and $n-k$ entries equal in $b$ where $b$ is not contained in any of the $c$-colorings, then $\bfy(\bfx) = \bfy(\bfx')$. The last case is an extension of the second case.
\end{proof}

Coloring profiles satisfying the conditions of Theorem \ref{thm:covering_pairs} in design theory are commonly known as $(q, c, 2)$-\textit{covering designs} which are defined in the following way.
\begin{definition}\label{def:t-design}
    Given a set of points $\mathcal{P}$ of cardinality $q$ and a set of blocks $\mathcal{B}$ consisting of $c$-subsets of $\mathcal{P}$. An incidence structure with respect to $\mathcal{P}$ and $\mathcal{B}$ is called a \textit{$(q, c, 2)$-covering design} if every $2$ distinct points are together incident with at least one block.
\end{definition}
This definition can be extended to general $(\nu, k, \tau)$-covering designs. The smallest number $\tmin(\nu, k, \tau)$ of blocks needed for a $(\nu, k, \tau)$-covering design to exist has been widely studied (see, for instance, \cite{gordon1995new,mills1992coverings,moore1896tactical} for an overview) and computed for many values of $\nu, k$ and $\tau$. Schönheim in \cite{schonheim1964coverings} gave the best lower bound given by
\begin{align}\label{eq:lowerbound_tmin}
    \tmin(\nu, k, \tau) \geq \ceil{\tfrac{\nu}{k}\ceil{\tfrac{\nu-1}{k-1} \cdots \ceil{\tfrac{\nu-\tau+1}{k-\tau+1}}}}.
\end{align}
Additionally, he showed that there exists $(\nu, k, \tau)$-covering design whose size this bound whenever, for any $h = 0, \ldots, \tau-1$, it holds
\begin{align}\label{eq:necessary_tmin}
    \textstyle\binom{\nu-h}{\tau-h} \big/ \binom{k-h}{\tau-l} \in \NN.
\end{align}
In our case, the bound in \eqref{eq:lowerbound_tmin} reduces to
\begin{align}\label{eq:lowerbound_tmin_ourcase}
    \tmin(q, c) \geq \ceil{\tfrac{q}{c}\ceil{\tfrac{q-1}{c-1}}}.
\end{align}
By the necessary condition \eqref{eq:necessary_tmin}, the bound is achieved for $\tmin(q, 2)$, i.e., we have $\tmin(q,2) = q(q-1)/2 = \binom{q}{2}$. Additionally, it was shown that the bound in \eqref{eq:lowerbound_tmin_ourcase} is met for $c = 3$ (\cite{moore1896tactical,reiss1859ueber}) and for $c = 4, 5$ (\cite{hanani1960quadruple}).
Upper bounds were later studied in, \cite{bluskov1998new,etzion1995bounds,gordon1995new,nurmela1993upper}. This covers the first question to Problem \ref{problem:tmin} which we will not further investigate in this paper.\\

Let us now focus on the second question stated in Problem~\ref{problem:tmin}. That is, given a $c$-coloring profile $\bfI$, where each coloring of $\bfI$ is drawn uniformly at random and without replacement from $\setcoloring(q, c)$, what is the maximum number $\Tmin(q,c)$ of draws such that any sequence $\bfx \in \alphabet_q^n$ is $\bfI$-reconstructible? 
Corollary \ref{cor:achieving_2_or_q-1} is an immediate consequence of Theorem \ref{thm:covering_pairs} and \eqref{eq:lowerbound_tmin_ourcase}.
\begin{corollary}\label{cor:achieving_2_or_q-1}
    Given $q \geq 3$, we have
    \begin{align}
        \Tmin(q, 2) &= \tmin(q, 2) = \binom{q}{2}, \quad \text{and} \\
        \Tmin(q, q-1) &= \tmin(q, q-1) = 3.        
    \end{align}
\end{corollary}

With the characterization of Theorem \ref{thm:covering_pairs} we cover the remaining cases $3 \leq c \leq q-2$ in Proposition \ref{prop:t_max}, and complete the answer to the second question of Problem \ref{problem:tmin}.
\begin{proposition}\label{prop:t_max}
    Let $q \geq 4$ and $c\geq 2$ be two positive integers. The minimum number of distinct $c$-coloring channels over a $q$-ary alphabet $\alphabet_q$ that allows to recover any input sequence is 
    \begin{align}\label{eu:prop:t_min}
        \Tmin(q, c) = \binom{q-1}{c} + \binom{q-2}{c-1} + 1 .
    \end{align}
\end{proposition}
\begin{proof}[Proof for $c=3$]
    We are looking for the minimum number $\Tmin(q, c)$ of distinct $c$-colorings drawn uniformly at random from $\setcoloring(q, c)$. By Theorem \ref{thm:covering_pairs}, $t$ distinct $c$-coloring channels over $\alphabet_q$ can reconstruct a sequence $\bfx \in \alphabet_q^n$ if and only if every pair $(a, b) \in \alphabet_q^2$ is contained in at least one of the $t$ colorings. To prove the statement, we draw one $c$-coloring after the other with a uniform distribution. At each step, we count how many new pairs have been added. This iteration is continued until all $\binom{q}{2}$ pairs are covered by at least one coloring. Due to random drawing, we need to consider the worst case scenario or pairs that are added at each step. For instance, any coloring we chose as first $c$-coloring $I_1$ covers $\binom{c}{2}= 3$ pairs in $\alphabet_q$. We now chose a second $c$-coloring $I_2$. In the worst case, we have $I_1 \cap I_2 = q-1$ or equivalently the number of new pairs added is $c-1$.

    Continuing this procedure, we observe that the worst case scenario consists in choosing all colorings of an alphabet $\mathcal{B} \in \alphabet_q$ of size $q-1$, that is $\card{\setcoloring(q-1, c)} = \binom{q-1}{c}$ many. At this point, only $q-1$ pairs are left to be covered and only $\binom{q-1}{c-1}$ colorings are left to be chosen all of which containing the element $\alpha := \alphabet_q \setminus \mathcal{B}$. Here we continue with the same idea, by adding all colorings that do not contain one last pair, say $(\alpha, \beta)$ for some $\beta \in \mathcal{B}$. Hence only colorings containing the elements $\alpha$ and $\beta$ are left. Adding any one of them covers the last pair and no more colorings are needed. To sum up, we have a worst case scenario of $\Tmin(q, c) = \binom{q-1}{c} + \binom{q-2}{c-1} + 1$ colorings.
\end{proof}

\section{Conclusion}\label{sec:concl}
We introduced an error-free channel model based on subsets of a given size $c$ over a $q$-ary alphabet, called $c$-colorings. This channel takes as an input a $q$-ary sequence $\bfx$ of given length $n$ and outputs subsequences of $\bfx$ consisting only of the symbols lying in the $c$-colorings considered. For given parameters $c, q$ and the number $t$ of distinct colorings considered, we discussed the information rate and capacity of the corresponding channel model. We proved that maximum information rate is achieved if and only if the $c$-colorings related to the channel form a $(q, c, 2)$-covering design. This allowed us to connect the minimum size for which a coloring profile achieves maximum information rate to the minimum covering number of a $(q, c, 2)$-covering design. Additionally, we gave an expression for the minimum size $\Tmin(q,c)$ for which any coloring profile of size $\Tmin(q,c)$ achieves maximum information rate.

\section*{Acknowledgment}
The authors thank Amit Meller for helpful discussions and introducing them with the problem of single protein nanopore sensing, which was the biological motivation to the information theoretical model studied in this work. 

This work was funded by the European Union (DiDAX, 101115134). Views and opinions expressed are however those of the authors only and do not necessarily reflect those of the European Union or the European Research Council Executive Agency. Neither the European Union nor the granting authority can be held responsible for them.

\bibliographystyle{IEEEtran}
\bibliography{biblio}

\end{document}